\documentclass[a4paper, 11pt]{article}

\usepackage{amsmath, amssymb, amsthm, multirow, url, fullpage}

\theoremstyle{plain}
\newtheorem{conjecture}{Conjecture}

\newtheorem{lemma}{Lemma}
\newtheorem{proposition}{Proposition}
\newtheorem{theorem}{Theorem}
\newtheorem{question}{Question}
\newtheorem*{claim*}{Claim}

\theoremstyle{definition}
\newtheorem{definition}{Definition}

\newcommand{\GF}{\mathrm{GF}}

\newcommand{\Sym}{\mathrm{Sym}}
\newcommand{\Alt}{\mathrm{Alt}}
\newcommand{\GL}{\mathrm{GL}}

\newcommand{\Aff}{\mathrm{Aff}}
\newcommand{\Cay}{\mathrm{Cay}}

\begin{document}

\title{Computing in permutation groups without memory}
\author{Peter J. Cameron\footnote{School of Mathematical Sciences, Queen Mary, University of London,  Mile End Road, London E1 4NS, UK and
	School of Mathematics and Statistics, University of St Andrews,	Mathematical Institute,	North Haugh, St Andrews, Fife KY16 9SS, UK.
   email: p.j.cameron@qmul.ac.uk, pjc@mcs.st-andrews.ac.uk},~Ben Fairbairn\footnote{Department of Economics, Mathematics and Statistics, Birkbeck, University of London, Malet Street, London WC1E 7HX, UK.	email: bfairbairn@ems.bbk.ac.uk}~and Maximilien Gadouleau\footnote{School of Engineering and Computing Sciences, Durham University, South Road, Durham, DH1 3LE, UK. email:m.r.gadouleau@durham.ac.uk}}
\maketitle

\begin{abstract}
Memoryless computation is a new technique to compute any function of a set of registers by updating one register at a time while using no memory. Its aim is to emulate how computations are performed in modern cores, since they typically involve updates of single registers. The memoryless computation model can be fully expressed in terms of transformation semigroups, or in the case of bijective functions, permutation groups. In this paper, we consider how efficiently permutations can be computed without memory. We determine the minimum number of basic updates required to compute any permutation, or any even permutation. The small number of required instructions shows that very small instruction sets could be encoded on cores to perform memoryless computation. We then start looking at a possible compromise between the size of the instruction set and the length of the resulting programs. We consider updates only involving a limited number of registers. In particular, we show that binary instructions are not enough to compute all permutations without memory when the alphabet size is even. These results, though expressed as properties of special generating sets of the symmetric or alternating groups, provide guidelines on the implementation of memoryless computation.
\end{abstract}


\section{Introduction}
\subsection{Memoryless computation}

Typically, swapping the contents of two variables $x$ and $y$ requires a buffer $t$, and proceeds as follows (using pseudo-code):
\begin{eqnarray*}
    t &\gets& x\\
    x &\gets& y\\
    y &\gets& t.
\end{eqnarray*}
However, a famous programming trick consists in using XOR (when $x$ and $y$ are sequences of bits), which we view in general as addition over a vector space:
\begin{eqnarray*}
    x &\gets& x+y\\
    y &\gets& x-y\\
    x &\gets& x-y.
\end{eqnarray*}
We thus perform the swap without any use of memory. 

While the example described above (commonly referred to as the XOR swap) is folklore in Computer Science, the idea to compute functions without memory was developed by Burckel et. al. in \cite{Bur96,Bur04,BGT09,BM00,BM04a,BM04} and was then independently rediscovered and expanded by Gadouleau and Riis in \cite{GR11a}. Amongst the results derived in the literature is the non-trivial fact that any function can be computed using memoryless computation. Moreover, only a number of updates linear in the number of registers is needed: any function of $n$ variables can be computed in at most $4n-3$ updates (a result proved in \cite{BGT09} for the case where variables are boolean and extended in \cite{GR11} and \cite{BGT13} independently for any finite alphabet), and only $2n-1$ updates if the function is bijective.

Memoryless computation has the potential to speed up computations in two ways. First, it avoids time-consuming communication with the memory, thus easing concurrent execution of different programs. Second, unlike traditional computing which treats registers as ``black boxes,'' memoryless computation effectively combines the values contained in those registers. Therefore, memoryless computation can be viewed as an analogue in computing to network coding \cite{ACLY00,YLCZ06}, an alternative to routing on networks. It is then proved in \cite{GR11a} that memoryless computation uses arbitrarily fewer updates than black-box computing for a certain class of manipulations of registers. 

\subsection{Model for computing in permutation groups without memory}

Let us recall some notation and results from \cite{GR11a}. Let $A$ be a finite set, referred to as the {\em alphabet}, of cardinality $q:= |A| \ge 2$ (usually, we assume $A = \mathbb{Z}_q$ or $A = \GF(q)$ if $q$ is a prime power). Let $n \ge 2$ be an integer representing the number of registers $x_1,\ldots,x_n$. We denote $[n] = \{1,2,\ldots,n\}$. The elements of $A^n$ are referred to as {\em states}, and any state $a \in A^n$ is expressed as $a= (a_1,\ldots,a_n)$, where $a_i$ is the $i$-th coordinate or register of $a$. For any $1 \le k \le n$, the $k$-th unit state is given by $e^k = (0,\ldots,0,1,0,\ldots,0)$ where the $1$ appears in coordinate $k$. We also denote the all-zero state as $e^0$.

Although the model in \cite{GR11a} considered the computation of any transformation of $A^n$, in this paper we only consider permutations of $A^n$. For any $f \in \Sym(A^n)$, we denote its $n$ coordinate functions as $f_1,\ldots,f_n : A^n \to A$, i.e.
$f(x) = (f_1(x), \ldots,f_n(x))$
for all $x = (x_1,\ldots,x_n) \in A^n$. We say that the $i$-th coordinate function is {\em trivial} if it coincides with that of the identity: $f_i(x) = x_i$; it is nontrivial otherwise.

An {\em instruction} is a permutation $g$ of $A^n$ with exactly one nontrivial coordinate function:
$$
    g(x) = (x_1,\ldots,x_{j-1},g_j(x),x_{j+1},\ldots,x_n).
$$
We say the instruction $g$ {\em updates} the $j$-th register. We can represent this instruction as
$$
    y_j \gets g_j(y)
$$
where $y = (y_1,\ldots,y_n) \in A^n$ represents the contents of the registers. By convention, we also let the identity be an instruction, but we shall usually omit it. We denote the set of all instructions in $\Sym(A^n)$ as $\mathcal{I}(A^n)$ (or simply $\mathcal{I}$ when there is no ambiguity). For instance, $\mathcal{I}(\GF(2)^2)$ is given by
\begin{equation} \label{eq:I22} \nonumber
    \mathcal{I} = \{ (x_1+1,x_2), (x_1+x_2,x_2), (x_1+x_2+1,x_2), (x_1,x_2+1), (x_1,x_1+x_2), (x_1,x_1+x_2+1) \}.
\end{equation}

A {\em program} computing $f \in \Sym(A^n)$ is a sequence $g^{(1)},\ldots,g^{(L)} \in \mathcal{I}$ such that
$$
    f = g^{(L)} \circ \cdots \circ g^{(1)}.
$$
In other words, a program computes $f$ by updating one register at a time, without any knowledge of the input. We remark that unless $f$ is the identity, we can assume that none of the instructions in its program is the identity; furthermore, we can always assume that $g^{(k+1)}$ updates a different register to $g^{(k)}$. The shortest length of a program computing $f$ is denoted as $\mathcal{L}(f)$ and referred to as the {\em complexity} of $f$. 

With this notation, the swap of two variables can be viewed as computing the permutation $f$ of $A^2$ defined as $f(x_1,x_2) = (x_2,x_1)$, and the program is given by
\begin{eqnarray*}
    y_1 &\gets& y_1 + y_2 \qquad (= x_1 + x_2)\\
    y_2 &\gets& y_1 - y_2 \qquad (= x_1)\\
    y_1 &\gets& y_1 - y_2 \qquad (= x_2).
\end{eqnarray*}
Thus, the complexity of the swap is three instructions.

Theorem 2.4 in \cite{GR11a} indicates that the instructions generate the symmetric group $\Sym(A^n)$: any permutation can be computed without memory. Once this is established, the first natural problem is to determine how fast permutations can be computed. Theorem 3.5 in \cite{GR11a} shows that the maximum complexity of any permutation of $A^n$ is exactly $2n-1$ instructions. Moreover, the average complexity is above $2n-3$ when the alphabet is binary and $n$ is large enough.

\subsection{Smaller instruction sets}

In this paper, we investigate another natural problem for memoryless computation. The model introduced above allows us to update a register by any possible function of all the registers. In practice the very large number of possible instructions makes it hard to encode all of them on a core. Therefore, we must search for limited instruction sets which are easy to encode while still salvaging the advantages offered by memoryless computation. Before answering this engineering problem, we will determine its theoretical limit, i.e. the size of the smallest instruction set which allows us to compute any function without memory. 

Once recast in algebraic language, the design of instruction sets able to compute any transformation becomes a problem on generating sets, which has been extensively studied for transformation semigroups \cite{Hig92,How95}. It is well known that to generate the full transformation semigroup, one only needs a generating set of the symmetric group and one more transformation \cite{GM09}. Since the last transformation can be an instruction \cite{GR11a}, we only consider the symmetric group in this paper. In Theorem \ref{th:sym_alt} (a), we prove that $\Sym(A^n)$ can be generated by only $n$ instructions (unless $q=n=2$, where three instructions are needed); we also prove a similar result for the alternating group in Theorem \ref{th:sym_alt} (b). This result illustrates that memoryless computation could be practically implemented on cores.

An ``efficient'' instruction set must satisfy the following tradeoff: it should contain a relatively small number of instructions and yet yield short programs. Also, the XOR swap is seldom used in practice, mostly because the instructions it involves cannot be easily pipelined. We thus expect a good set of instructions to offer the possibility to easily pipeline instructions. We will then investigate a natural candidate for a possible generating set of instructions. We will try to compute functions using only ``local'' instructions, which only involve a limited number of registers, i.e. {\em $l$-ary} instructions.

\begin{definition} \label{def:binary}
A coordinate function $f_j: A^n \to A$ is $l$-ary if it only involves at most $l$ variables:
$$
    f_j(x) = f_j(x_{k_1},\ldots,x_{k_l})
$$
for some $k_1,\ldots,k_l \in [n]$. For $l=1,2$, we say it is unary, binary respectively. A permutation whose coordinate functions are all $l$-ary is also referred to as $l$-ary.
\end{definition}

Binary Boolean instructions are not sufficient to compute any Boolean function \cite{GR11a}. In Theorem \ref{th:l-ary}, we settle the general case. In particular, an $n$-ary instruction is required to generate $\Sym(A^n)$ when $|A|$ is even.

Another type of search for a trade-off between small instruction sets and short programs is investigated in \cite{CFG12a}, where we compute linear functions only using linear updates. Analogous results to Theorem \ref{th:sym_alt} and to the maximum complexity in \cite{GR11a} are determined for the general and special linear groups.

The rest of the paper is organised as follows. Section \ref{sec:preliminaries} gives some preliminary definitions and results. Among others, it proves that any even permutation can be computed by even instructions in most cases. Section \ref{sec:smallest} determines the minimum size of a generating set of instructions for the symmetric and alternating groups. Then, in Section \ref{sec:l-ary} we determine the groups generated by $l$-ary instructions. Finally, Section \ref{sec:open} gives some open questions and perspectives on the topic.

\section{Preliminaries} \label{sec:preliminaries}

\subsection{Action on the set of instructions by conjugation}

We remark that the set of $l$-ary permutations forms a group if and only if $l \in \{1,n\}$. Indeed, for $l=1$ the unary permutations form the group 
$$
	U := \Sym(A)\, \mbox{Wr}\, \Sym(n)
$$ 
and for $l=n$ they form $\Sym(A^n)$. Conversely, if $2 \le l \le n-1$, then it is clear that the following program uses $l$-ary instructions and yet does not compute an $l$-ary permutation:
\begin{align*}
    y_2 &\gets y_2 + y_{l+1}\\
    y_1 &\gets \sum_{i=1}^l y_i.
\end{align*}

A {\em permutation of variables} is a permutation $f^\pi$ of $A^n$ such that
$f^\pi(x) = (x_{1\pi},\ldots,x_{n\pi})$
for some $\pi \in \Sym(n)$; note that $(f^\pi)^{-1} = f^{\pi^{-1}}$. In other words, the permutations of variables represent an action of $\Sym(n)$ on $A^n$.

\begin{proposition} \label{prop:U}
The largest group acting by conjugation on the set of instructions $\mathcal{I}$ is the group $U$ of unary permutations.
\end{proposition}

\begin{proof}
First, let us show that $U$ acts on $\mathcal{I}$ by conjugation. Note that $U$ is generated by the permutations of variables and the unary instructions. Let $h \in \Sym(A^n)$ be a unary instruction with nontrivial coordinate function $h_i(x_i)$, where $h_i \in \Sym(A)$; then $h^{-1}(x) = (x_1,\ldots,h_i^{-1}(x_i),\ldots,x_n)$. Let $g$ be an instruction updating register $j$, then
$$
    h^{-1}gh(x) = \begin{cases} (x_1,\ldots,g_j(h(x)),\ldots,x_n) & \mbox{if } i \ne j,\\
    (x_1,\ldots,h_i^{-1}(g_i(h(x))),\ldots,x_n) & \mbox{otherwise}. \end{cases}
$$
In both cases, it is an instruction updating register $j$.

Let $f^\pi$ be a permutation of variables, then it is easily checked that
$$
    f^{\pi^{-1}} g f^\pi(x) = (x_1,\ldots,g_j(f^\pi(x)), \ldots, x_n),
$$
where the nontrivial term appears in coordinate $j\pi^{-1}$.

Second, let us prove that any non-unary permutation does not act on $\mathcal{I}$ by conjugation. Let $f \in \Sym(A^n)$ have a non-unary coordinate function (and hence $f^{-1}$ also), then there is an $x_i$ which appears more than once in the coordinate functions of $f^{-1}$, say in $f^{-1}_1(x_i,\ldots)$ and $f^{-1}_2(x_i,\ldots)$. Then there exist two pairs of states $a^k,b^k$ ($k = 1,2$) in $A^n$ only differing in register $i$ such that $f^{-1}_k(a^k) \ne f^{-1}_k(b^k)$. Let $g$ be an instruction updating the $i$-th register such that
$g(a^1) = b^1, g(b^2) = a^2.$ Then denoting $u = f^{-1}(a^1)$ and $v = f^{-1}(b^2)$, we obtain
\begin{eqnarray*}
    f^{-1}_1 g f(u) &=& f^{-1}_1 g(a^1) = f^{-1}_1(b^1) \ne u_1\\
    f^{-1}_2 g f(v) &=& f^{-1}_2 g(b^2) = f^{-1}_2(a^2) \ne v_2.
\end{eqnarray*}
Therefore, $f^{-1} g f$ updates the first and second registers and as such is not an instruction.
\end{proof}

We remark that if $g,g'$ are $U$-conjugates (i.e. $g = hgh^{-1}$ for some $h \in U$), then $\mathcal{L}(g) = \mathcal{L}(g')$.

\subsection{Internally computable permutation groups}

Let $G \le \Sym(A^n)$ and $g \in G$. We say that $g$ is {\em computable} in $G$ if there exists a program computing $g$ consisting only of instructions from $G$. The set of elements computable in $G$ is hence given by the subgroup $\langle G \cap \mathcal{I} \rangle$. We say $G$ is {\em internally computable} if all its elements are computable therein, i.e. if $G = \langle G \cap \mathcal{I} \rangle$.

In order to illustrate how subtle the question of internal computability is, let us consider cyclic permutation groups. If $g$ is an instruction, then clearly $\langle g \rangle$ is an internally computable group. On the other hand, consider $\Sym(\{0,1,2\}^2)$; for convenience number the elements of $\{0,1,2\}^2$ in Gray code as $1:=00$, $2:=01$, $3:=02$, $4:=12$, $5:=11$, $6:=10$, $7:=20$, $8:=21$, $9:=22$ and let $g:=(123)(67)$. The only generating sets of size one for the group $\langle g \rangle$ are $g$ and $g^{-1}$ and clearly neither of these is an instruction. Simultaneously, the elements $g^2$ and $g^3$ are instructions and since $g=g^3(g^2)^{-1}$ it follows that $\langle g \rangle = \langle g^2,g^3 \rangle$ so the group $\langle g \rangle$ is indeed internally computable, despite the fact that $g$ is not an instruction.

The problem of determining whether a permutation group is internally computable can reveal some surprises. For instance, the alternating group $\Alt(\GF(2)^2)$ is not internally computable. Indeed, recall that the set of instructions is given by
$$
    \mathcal{I} = \{ (x_1+1,x_2), (x_1+x_2,x_2), (x_1+x_2+1,x_2), (x_1,x_2+1), (x_1,x_1+x_2), (x_1,x_1+x_2+1) \}.
$$
Only the instructions $(x_1+1,x_2)$ and $(x_1,x_2+1)$ are even; those instructions only generate a group of order $4$.

\begin{proposition}
The alternating group $\Alt(A^n)$ is internally computable unless $q=2$ and $n=2$ or $n=3$.
\end{proposition}

\begin{proof}
The proof for $q=2$ will follow Theorem \ref{th:l-ary} (the case for $q=2$, $n=3$ is checked by computer). For $q \ge 3$, let $G$ be the group generated by even instructions. Define a relation $\sim_1$ on $A^n$ by $x\sim_1 y$ if $x=y$ or there exists $z$ such
that the $3$-cycle $(x,y,z)$ is in $G$. This is an equivalence
relation, and it has the property that, if $x\sim_1 y$ and $y\sim_1 z$, then
$(x,y,z)\in G$. So if $\sim_1$ is the universal relation, then $G$ contains
every $3$-cycle, and is the alternating group.

Let $x,y$ differ in one position, i.e. $y = x + \lambda e^i$ for some $\lambda \ne 0$ and some $1 \le i \le n$. Then for any $\mu \notin \{0,\lambda\}$, $(x,y,z)$ is an even instruction, where $z = x + \mu e^i$. Thus any two states are equivalent when they differ in one register, and hence $\sim_1$ is the universal relation.
\end{proof}

Given an internally computable group $G$, one may ask two ``extreme'' questions, similar to the problems for the symmetric group.
\begin{enumerate}
    \item What is the maximum complexity of an element in $G$, i.e. the diameter of the Cayley graph $\Cay(G, G \cap \mathcal{I})$? This indicates how fast we can compute any element in $G$ if we allow any instruction.

    \item What is the smallest cardinality of a set of instructions in $G$ which generates $G$? This indicates the minimum amount of space required to store the instructions needed to compute any element of $G$.
\end{enumerate}

\subsection{Fast permutations}

For any set of instructions $\mathcal{J} \subseteq \mathcal{I}$ and any $g \in \langle \mathcal{J} \rangle$, we denote the shortest length of a program computing $g$ using only instructions from $\mathcal{J}$ as $\mathcal{L}(g,\mathcal{J})$. (Note that we only look at instructions in $\mathcal{J}$, and not all instructions in $\langle \mathcal{J} \rangle$). If $\mathcal{J} \subseteq \mathcal{K}$, we say $g$ is $(\mathcal{J},\mathcal{K})$-{\em fast} if $\mathcal{L}(g,\mathcal{J}) = \mathcal{L}(g,\mathcal{K})$.

We have the following properties: let $\mathcal{J} \subseteq \mathcal{K} \subseteq \mathcal{M}$ and $g,h \in \langle \mathcal{J} \rangle$.
\begin{enumerate}
    \item If $\mathcal{J}$ is symmetric and $g$ is $(\mathcal{J},\mathcal{K})$-fast, then so is $g^{-1}$.

    \item If $g,h$ are $(\mathcal{J},\mathcal{K})$-fast and $\mathcal{L}(gh,\mathcal{K}) = \mathcal{L}(g,\mathcal{K}) + \mathcal{L}(h,\mathcal{K})$, then $gh$ is also $(\mathcal{J},\mathcal{K})$-fast.

    \item If $g$ is $(\mathcal{J},\mathcal{M})$-fast, then $g$ is $(\mathcal{J},\mathcal{K})$-fast.

    \item If $g$ is $(\mathcal{J},\mathcal{K})$-fast and $(\mathcal{K},\mathcal{M})$-fast, then $g$ is $(\mathcal{J},\mathcal{M})$-fast.
\end{enumerate}

If all elements of an internally computable group $K$ are $(K \cap \mathcal{I}, G \cap \mathcal{I})$-fast for some $K \le G$, we say that $K$ is fast in $G$. (Typically, $G = \Sym(A^n)$). For instance, if $K$ is the group of all instructions updating a given register, then $K$ is clearly fast. On the other hand, the alternating group is never fast.

\begin{proposition} \label{prop:Alt_not_fast}
The alternating group $\Alt(A^n)$ is not fast for any $A$ and $n$.
\end{proposition}

\begin{proof}
Let $g = (e^0,e^1,e^2) = (e^0,e^2) \circ (e^0,e^1)$ be an even permutation of $A^n$. Then $\mathcal{L}(g,\Sym(A^n)) = 2$ since both transpositions are instructions. On the other hand, any program computing $g$ of length $2$ must begin with either its first or second coordinate function $g_1$ or $g_2$, i.e.
\begin{eqnarray*}
    \mbox{either}\quad y_1 \gets g_1(y) &=& y_1 + \delta(y,e^0) - \delta(y,e^1)\\
    \mbox{or}\quad y_2 \gets g_2(y) &=& y_2 + \delta(y,e^1) - \delta(y,e^2),
\end{eqnarray*}
where $\delta$ is the Kronecker delta function. However, the first corresponding instruction is the transposition $(e^0,e^1)$, while the second is not even a permutation, for it maps both $e^2$ and $e^0$ to $e^0$.
\end{proof}

The problem of determining whether a group $G$ is fast has three important special cases.
\begin{enumerate}
    \item $G = \GL(n,q)$ for $q$ a prime number. This indicates whether one can compute linear functions any faster by allowing nonlinear instructions.
	\begin{conjecture}
	$\GL(n,q)$ is fast in $\Sym(\GF(q)^n)$.
	\end{conjecture}        
	Two partial results are already known. First, Theorem 4.7 in \cite{GR11a} shows that the permutation matrices are fast in the general linear group: it takes exactly $n-F+C$ instructions to compute a permutation of variables with $F$ fixed points and $C$ cycles, and this can be done via linear instructions. Secondly, Proposition 6 in \cite{GR11a} shows that for large $q$, almost all of $\GL(n,q)$ can be computed in $n$ linear instructions and hence almost all of the general linear group is fast.

    \item $G = \Sym(A^k) \times \Sym(A^{n-k})$ acting coordinatewise. The significance of this group can be explained as follows. Suppose we want to compute a function of $k$ registers only, but we have $n$ registers available. The additional $n-k$ registers can then be used as additional memory. It is known that using additional memory can yield shorter programs in some cases \cite{GR11a}. However, all the shortest programs using memory known so far ``erase'' the memory content and replace it with functions of the first $k$ registers. If $G$ is fast, then one cannot compute the original function of $k$ registers any faster without erasing some knowledge of the last $n-k$ registers. More formally, $G$ is fast if and only if the following conjecture is true.

    \begin{conjecture}
	Let $g \in \Sym(A^k)$, $h \in \Sym(A^{n-k})$ and define $f \in \Sym(A^n)$ by $(f_1(x),\ldots,f_k(x)) = g(x_1,\ldots,x_k)$ and $(f_{k+1}(x),\ldots,f_n(x)) = h(x_{k+1},\ldots,x_n)$. Then
        $$
            \mathcal{L}(f) = \mathcal{L}(g) + \mathcal{L}(h).
        $$
	\end{conjecture}
	
    \item $G = \Sym(A^k)$ acting on the first $k$ coordinates. This might be a simpler case than the previous one.
\end{enumerate}

\section{Smallest sets of generating instructions} \label{sec:smallest}

We now turn to the problem of determining the smallest number of instructions generating the symmetric group and the alternating group. Clearly, one needs to update all $n$ registers to generate a transitive group.

\begin{theorem} \label{th:sym_alt}
\begin{enumerate}[(a)]
	\item Unless $q=n=2$, $\Sym(A^n)$ is generated by $n$ instructions.
	
	\item If $q \ge 3$, then $\Alt(A^n)$ is generated by $n$ instructions.
\end{enumerate}
\end{theorem}

We recall a classical theorem of Jordan (see for instance
\cite[Theorem 3.3E]{DM96}).

\begin{lemma}
Let $G\leq\Sym(m)$ be primitive and suppose that $G$ contains
a cycle of length $p$ for some prime $p\leq m-3$. Then
$G=\Sym(m)$ or $\Alt(m)$.
\end{lemma}

We first deal with small alphabets.

\begin{lemma}\label{thing}
The group $\Alt(\{0,1,2\}^n)$ is generated by $n$ instructions.
\end{lemma}

\begin{proof}
The case $n=2$ is easily dealt with separately, so we shall assume
that $n>2$.

We order the elements of $A^n$ lexicographically and number them
accordingly (so for instance if $n=3$ then the elements in order are
000=:1, 001=:2, 002=:3, 010=:4, 011=:5,$\ldots$). We define the
permutations $\pi_1,\ldots,\pi_n$ as follows. First $\pi_1=(1,2,3)$,
so $\pi_i$ only updates the rightmost register. For $2\leq i\leq n$
the permutation $\pi_i$ is the unique instruction updating the
$i^{th}$ register that is a product of $3^{n-1}$ cycles of length 3
all but the last of which, when written in the usual cycle notation,
lists states in lexicographic order. For example, if $n=3$ these
permutations are
\begin{align*}
    \pi_1 &= (1,2,3),\\
    \pi_2 &= (1,4,7)(2,5,8)(3,6,9)(10,13,16)(11,14,17)\\
    & \hspace{20mm} (12,15,18)(19,22,25)(20,23,26)(27,24,21) \,\mbox{and}\\
    \pi_3 &= (1,10,19)(2,11,20)(3,12,21)(4,13,22)(5,14,23)\\
    & \hspace{20mm} (6,15,24)(7,16,25)(8,17,26)(27,18,9).
\end{align*}

We claim that these permutations generate a 2-transitive group. It
is easy to see that they generate a transitive group. We suppose
$n\geq4$ and proceed by induction, the case $n=3$ be easily verified
by computer. By hypothesis the group generated by
$\pi_1,\ldots,\pi_{n-1}$ gives all permutations of the points the
points $\{1,\ldots,q^{n-1}\}$. In particular we have the
permutations $(i,i+1,i+2)$ for $1\leq i\leq q^{n-1}-3$ and thus the
permutations $(i,i+1,i+1)^{\pi_n^j}$ for $1\leq i\leq q^{n-1}-4$ and
$1\leq j\leq q-1$. Furthermore, we also have the permutations
$(1,2,q^{n-1})^{\pi_n^j}$ for $1\leq j\leq q-1$ which altogether
gives us enough permutations to act transitively on the stabilizer
of any point in $\{1,\ldots,q^{n-1}\}$.

We have shown that $\pi_1,\ldots,\pi_n$ generate a group that is
2-transitive and therefore primitive. Since $\pi_1$ is a 3-cycle we
can now apply Jordan's Theorem to conclude that $\pi_1,\ldots,\pi_n$
generates the whole of Alt($A^n$).
\end{proof}

\begin{lemma} \label{lemma:sym_2}
The group $\Sym(\{0,1\}^n)$ is generated by $n$ instructions.
\end{lemma}

\begin{proof}
We define a set of permutations $\pi_1,\ldots,\pi_n$ as follows. We
order words in $\{0,1\}^n$ in the usual Gray code ordering and label
them $1,\ldots,2^n$ (so for instance if $n=3$ the elements in order are $000=: 1$, $001 =: 2$, $011 =:3$, $010 =: 4$, $110 =: 5$). We now define $\pi_1:=(1,2)$. For $2\leq i\leq
n$ we construct $\pi_i$ by taking the derangement that updates the
$i^{th}$ register and removing the first cycle that interchanges two
adjacent words. For example, if $n=3$ then
\begin{align*}
    \pi_1 &= (1,2)\\
    \pi_2 &= (1,4)(5,8)(6,7) \,\mbox{and}\\
    \pi_3 &= (1,8)(2,7)(3,6).
\end{align*}
A straightforward induction analogous to the proof of Lemma \ref{thing}
now enables us to show that the above generate a primitive group
containing a transposition and so we apply Jordan's Lemma to show
that the above generate the whole of $\Sym(A^n)$.
\end{proof}

\begin{lemma}
If $q>2$ is odd then Sym($A^n$) is generated by $n$ instructions.
\end{lemma}

\begin{proof}
We define the permutation $\pi_1:A^n\rightarrow A^n$ as follows. For
$a:=(a_1,a_2,\ldots,a_n)$ we have that
$$
\pi_1:a\mapsto\left\{\begin{array}{ll}
(1-a_1,0,\ldots,0)&\mbox{if }a_1\in\{0,1\}\mbox{ and }a_2=\cdots=a_n=0\\
a&\mbox{if }a_1>1\mbox{ and }a_2=a_3=\cdots=a_n=0\\
(0,a_2,\ldots,a_n)&\mbox{if }a_1=q-1\mbox{ and }a_i\not=0\mbox{ for
some }2\leq i\leq n\\
(a_1+1,a_2,\ldots,a_n)&\mbox{otherwise.}\\
\end{array}\right.
$$
For $2\leq r\leq n$ we define the permutation $\pi_r:A^n\rightarrow
A^n$ as follows.
$$
\pi_r:a\mapsto\left\{\begin{array}{ll}
(a_1,\ldots,a_{r-1},0,a_{r+1},\ldots,a_n)&\mbox{if }a_r=q-1\mbox{
and }a_i\not=q-1\mbox{ for some }i\not=r\\
(a_1,\ldots,a_{r-1},a_r+1,a_{r+1},\ldots,a_n)&\mbox{if
}a_r\not=q-1\mbox{ and }a_i\not=q-1\mbox{ for some }i\not=r\\
(a_1,\ldots,a_{r-1},q-1,a_{r+1},\ldots,a_n)&\mbox{if }a_r=0\mbox{
and }a_i=q-1\mbox{ for all }i\not=r\\
(a_1,\ldots,a_{r-1},a_r-1,a_{r+1},\ldots,a_n)&\mbox{otherwise.}
\end{array}\right.
$$

For example, if $q=3$ and $n=3$ then, writing the elements of $A^n$
as a series of grids, we have that the permutation $\pi_1$ is

\hspace{30mm} \setlength{\unitlength}{5mm}
\begin{picture}(0,2)
\put(-0.5,0.5){$a=(a_1,a_2,0)$} \put(4.5,0.5){$a=(a_1,a_2,1)$}
\put(9.5,0.5){$a=(a_1,a_2,2)$}

\multiput(0,0)(0,-1){4}{\line(1,0){3}}
\multiput(0,0)(1,0){4}{\line(0,-1){3}} \put(0.5,-0.5){\line(1,0){1}}
\multiput(0.5,-1.5)(0,-1){2}{\vector(1,0){2}}

\multiput(5,0)(0,-1){4}{\line(1,0){3}}
\multiput(5,0)(1,0){4}{\line(0,-1){3}}
\multiput(5.5,-0.5)(0,-1){3}{\vector(1,0){2}}

\multiput(10,0)(0,-1){4}{\line(1,0){3}}
\multiput(10,0)(1,0){4}{\line(0,-1){3}}
\multiput(10.5,-0.5)(0,-1){3}{\vector(1,0){2}}
\end{picture}
\vspace{20mm}

\noindent whilst the permutation $\pi_2$ is

\hspace{30mm} \setlength{\unitlength}{5mm}
\begin{picture}(0,2)
\put(-0.5,0.5){$a=(a_1,a_2,0)$} \put(4.5,0.5){$a=(a_1,a_2,1)$}
\put(9.5,0.5){$a=(a_1,a_2,2)$}

\multiput(0,0)(0,-1){4}{\line(1,0){3}}
\multiput(0,0)(1,0){4}{\line(0,-1){3}}
\multiput(0.5,-0.5)(1,0){3}{\vector(0,-1){2}}

\multiput(5,0)(0,-1){4}{\line(1,0){3}}
\multiput(5,0)(1,0){4}{\line(0,-1){3}}
\multiput(5.5,-0.5)(1,0){3}{\vector(0,-1){2}}

\multiput(10,0)(0,-1){4}{\line(1,0){3}}
\multiput(10,0)(1,0){4}{\line(0,-1){3}}
\multiput(10.5,-0.5)(1,0){2}{\vector(0,-1){2}}
\put(12.5,-2.5){\vector(0,1){2}}
\end{picture}
\vspace{20mm}

\noindent and similarly for $\pi_3$.

We claim that $G:=\langle \pi_1,\ldots,\pi_n\rangle$ is
2-transitive. It is easy to see that $G$ acts transitively on $A^n$.
We will show that the stabilizer of the state $(q-1,q-1,\ldots,q-1)$ is
transitive on the remaining states.

Note that every cycle of $\pi_1$ apart from one has length $q$,
which is odd, the remaining cycle being a transposition. It follows
that $\tau:=\pi_1^{q-1}$ is a transposition. It is easy to see that
repeatedly conjugating $\tau$ by the various $\pi_i$s we have enough
elements for the stabilizer of $(q-1,q-1,\ldots,q-1)$ to act transitively
on the remaining points, that is, $G$ acts 2-transitively.

Since any 2-transitive action is primitive it follows that $G$ acts
primitively and since $\tau\in G$ is a transposition, Jordan's
Theorem tells us that $G=\mbox{Sym}(A^n)$.
\end{proof}

\begin{lemma}
If $q>2$ is even then Sym($A^n$) is generated by $n$ instructions.
\end{lemma}

\begin{proof}
We define the permutation $\pi_1:A^n\rightarrow A^n$ as follows. For
$a:=(a_1,a_2,\ldots,a_n)$ we have that
$$
\pi_1:a\mapsto\left\{\begin{array}{ll}
(1-a_1,0,\ldots,0)&\mbox{if }a_1\in\{0,1\}\mbox{ and }a_2=\cdots=a_n=0\\
a&\mbox{if }a_1>1\mbox{ and }a_2=a_3=\cdots=a_n=0\\
a&\mbox{if }a_1=0\mbox{ and }a_i\not=0\mbox{ for
some }2\leq i\leq n\\
(1,a_2,\ldots,a_n)&\mbox{if }a_1=q-1\mbox{ and }a_i\not=0\mbox{ for
some }2\leq i\leq n\\
(a_1+1,a_2,\ldots,a_n)&\mbox{otherwise.}\\
\end{array}\right.
$$
For $2\leq r\leq n$ we define the permutation $\pi_r:A^n\rightarrow
A^n$ as follows.
$$
\pi_r:a\mapsto\left\{\begin{array}{ll}
(a_1,\ldots,a_{r-1},a_r+1,a_{r+1},\ldots,a_n)&\mbox{if
}a_r\not=q-1\mbox{ and }a_i\not=q-1\mbox{ for some }i\not=r\\
(a_1,\ldots,a_{r-1},a_r+1,a_{r+1},\ldots,a_n)&\mbox{if
}a_r<q-2\mbox{ and }a_i=q-1\mbox{ for all }i\not=r\\
(a_1,\ldots,a_{r-1},0,a_{r+1},\ldots,a_n)&\mbox{if
}a_r=q-1\mbox{ and }a_i\not=q-1\mbox{ for some }i\not=r\\
(a_1,\ldots,a_{r-1},0,a_{r+1},\ldots,a_n)&\mbox{if
}a_r=q-2\mbox{ and }a_i=q-1\mbox{ for all }i\not=r\\
a&\mbox{if }a_1=a_2=\ldots=a_n=q-1\\
\end{array}\right.
$$

For example, if $q=4$ and $n=3$ then, writing the elements of $A^n$
as a series of grids, we have that the permutation $\pi_1$ is

\hspace{10mm} \setlength{\unitlength}{5mm}
\begin{picture}(0,2)
\put(0.0,0.5){$a=(a_1,a_2,0,0)$} \put(5.0,0.5){$a=(a_1,a_2,1,*)$}
\put(10.0,0.5){$a=(a_1,a_2,2,*)$} \put(15.0,0.5){$a=(a_1,a_2,3,*)$}

\multiput(0,0)(0,-1){5}{\line(1,0){4}}
\multiput(0,0)(1,0){5}{\line(0,-1){4}} \put(0.5,-0.5){\line(1,0){1}}
\multiput(1.5,-1.5)(0,-1){3}{\vector(1,0){2}}

\multiput(5,0)(0,-1){5}{\line(1,0){4}}
\multiput(5,0)(1,0){5}{\line(0,-1){4}}
\multiput(6.5,-0.5)(0,-1){4}{\vector(1,0){2}}

\multiput(10,0)(0,-1){5}{\line(1,0){4}}
\multiput(10,0)(1,0){5}{\line(0,-1){4}}
\multiput(11.5,-0.5)(0,-1){4}{\vector(1,0){2}}

\multiput(15,0)(0,-1){5}{\line(1,0){4}}
\multiput(15,0)(1,0){5}{\line(0,-1){4}}
\multiput(16.5,-0.5)(0,-1){4}{\vector(1,0){2}}
\end{picture}
\vspace{25mm}

\noindent where the star means that $a_4$ can take any value, whilst the permutation $\pi_2$ is

\hspace{10mm} \setlength{\unitlength}{5mm}
\begin{picture}(0,2)
\put(0.0,0.5){$a=(a_1,a_2,0,*)$} \put(5.0,0.5){$a=(a_1,a_2,1,*)$}
\put(10.0,0.5){$a=(a_1,a_2,2,*)$} \put(15.0,0.5){$a=(a_1,a_2,3,*)$}

\multiput(0,0)(0,-1){5}{\line(1,0){4}}
\multiput(0,0)(1,0){5}{\line(0,-1){4}}
\multiput(0.5,-0.5)(1,0){4}{\vector(0,-1){3}}

\multiput(5,0)(0,-1){5}{\line(1,0){4}}
\multiput(5,0)(1,0){5}{\line(0,-1){4}}
\multiput(5.5,-0.5)(1,0){4}{\vector(0,-1){3}}

\multiput(10,0)(0,-1){5}{\line(1,0){4}}
\multiput(10,0)(1,0){5}{\line(0,-1){4}}
\multiput(10.5,-0.5)(1,0){4}{\vector(0,-1){3}}

\multiput(15,0)(0,-1){5}{\line(1,0){4}}
\multiput(15,0)(1,0){5}{\line(0,-1){4}}
\multiput(15.5,-0.5)(1,0){3}{\vector(0,-1){3}}
\put(18.5,-0.5){\vector(0,-1){2}}

\end{picture}
\vspace{25mm}

\noindent and similarly for $\pi_3$ and $\pi_4$.

The argument concludes in the same way as the previous lemma: these
permutations generate a group that is 2-transitive and thus
primitive and contains a transposition, so by Jordan's Theorem our
result follows.
\end{proof}

We remark that $q=2$ must naturally be handled separately since all
cycles of every instruction in that case have length two.

More generally we ask the following.

\begin{question}
What do minimum generating sets of instructions look like?
\end{question}

As a partial answer to this question we have the following.

\begin{lemma} 
Let $X\subset\mbox{Sym}(A^n)$ be a set of $n$ instructions which generates $\Sym(A^n)$.
Then $X$ does not contain a unary instruction.
\end{lemma}

\begin{proof}
Let $\sim_r$ be the equivalence relation on $A^n$ where $x \sim_r y$ if and only if $x_r = y_r$. Then any instruction updating any register other than $r$ preserves $\sim_r$. Moreover, any unary instruction updating the register $r$ also preserves $\sim_r$. Therefore, if $X$ contains a unary instruction, it preserves $\sim_r$ for some $r$ and hence cannot generate $\Sym(A^n)$. 
\end{proof}

{\bf Comment} Since we can order the states in $A^n$ such that two consecutive  states only differ in one register (it is called a $(q,n)$-Gray code \cite{Gua98}), the Coxeter generators according to that ordering are all instructions. Therefore, the maximum size of a minimal set of generating instructions (i.e., a generating set whose proper subsets are not generating) is exactly $q^n-1$.\\
\\
We proceed to discuss Theorem \ref{th:sym_alt} (b).
Essentially our argument is the same as the previous two lemmas
replacing each of our transpositions with 3-cycles and resorting to
Jordan's Theorem to deduce the final result. Unfortunately, in this
case we need to split off into more cases than simply even and odd
since we now need to keep track not only of the length of the cycles
mod 3 but also of the parities of the permutations to ensure that
the permutations we construct are in fact all even.

\begin{lemma}\label{AltLemma}
If $q\equiv$ 1 or 5 ($mod$ 6) and $q>1$ then Alt($A^n$) is generated
by $n$ instructions.
\end{lemma}

\begin{proof}
We define the permutation $\pi_1:A^n\rightarrow A^n$ as follows. For
$a:=(a_1,a_2,\ldots,a_n)$ we have that
$$
\pi_1:a\mapsto\left\{\begin{array}{ll}
(1,0,\ldots,0)&\mbox{if }a_1=a_2=\cdots=a_n=0\\
(2,0,\ldots,0)&\mbox{if }a_1=1\mbox{ and }a_2=\cdots=a_n=0\\
(0,0,\ldots,0)&\mbox{if }a_1=2\mbox{ and }a_2=\cdots=a_n=0\\
a&\mbox{if }a_1>2\mbox{ and }a_2=\cdots=a_n=0\\
(0,a_2,\ldots,a_n)&\mbox{if }a_1=q-1\mbox{ and }a_i\not=0\mbox{ for
some }2\leq i\leq n\\
(a_1+1,a_2,\ldots,a_n)&\mbox{otherwise.}\\
\end{array}\right.
$$
For $2\leq r\leq n$ we define the permutation $\pi_r:A^n\rightarrow
A^n$ as follows.
$$
\pi_r:a\mapsto\left\{\begin{array}{ll}
(a_1,\ldots,a_{r-1},0,a_{r+1},\ldots,a_n)&\mbox{if }a_r=q-1\mbox{
and }a_i\not=q-1\mbox{ for some }i\not=r\\
(a_1,\ldots,a_{r-1},a_r+1,a_{r+1},\ldots,a_n)&\mbox{if
}a_r\not=q-1\mbox{ and }a_i\not=q-1\mbox{ for some }i\not=r\\
(a_1,\ldots,a_{r-1},q-1,a_{r+1},\ldots,a_n)&\mbox{if }a_r=0\mbox{
and }a_i=q-1\mbox{ for all }i\not=r\\
(a_1,\ldots,a_{r-1},a_r-1,a_{r+1},\ldots,a_n)&\mbox{otherwise.}\\
\end{array}\right.
$$

For example, if $q=5$ and $n=3$ then, writing the elements of $A^n$
as a series of grids, we have that the permutation $\pi_1$ is
\newpage

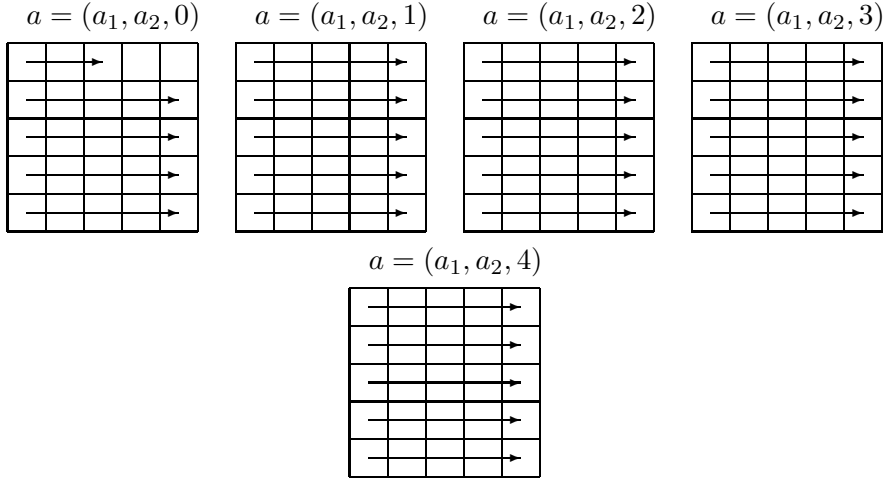
\setlength{\unitlength}{5mm}
\begin{picture}(0,2)
\put(0.5,0.5){$a=(a_1,a_2,0)$} \put(6.5,0.5){$a=(a_1,a_2,1)$}
\put(12.5,0.5){$a=(a_1,a_2,2)$} \put(18.5,0.5){$a=(a_1,a_2,3)$}
\put(9.5,-6){$a=(a_1,a_2,4)$}

\multiput(0,0)(0,-1){6}{\line(1,0){5}}
\multiput(0,0)(1,0){6}{\line(0,-1){5}}
\put(0.5,-0.5){\vector(1,0){2}}
\multiput(0.5,-1.5)(0,-1){4}{\vector(1,0){4}}

\multiput(6,0)(0,-1){6}{\line(1,0){5}}
\multiput(6,0)(1,0){6}{\line(0,-1){5}}
\multiput(6.5,-0.5)(0,-1){5}{\vector(1,0){4}}

\multiput(12,0)(0,-1){6}{\line(1,0){5}}
\multiput(12,0)(1,0){6}{\line(0,-1){5}}
\multiput(12.5,-0.5)(0,-1){5}{\vector(1,0){4}}

\multiput(18,0)(0,-1){6}{\line(1,0){5}}
\multiput(18,0)(1,0){6}{\line(0,-1){5}}
\multiput(18.5,-0.5)(0,-1){5}{\vector(1,0){4}}

\multiput(9,-6.5)(0,-1){6}{\line(1,0){5}}
\multiput(9,-6.5)(1,0){6}{\line(0,-1){5}}
\multiput(9.5,-7.0)(0,-1){5}{\vector(1,0){4}}
\end{picture}
\vspace{58mm}

\noindent whilst the permutation $\pi_2$ is \vspace{-2mm}

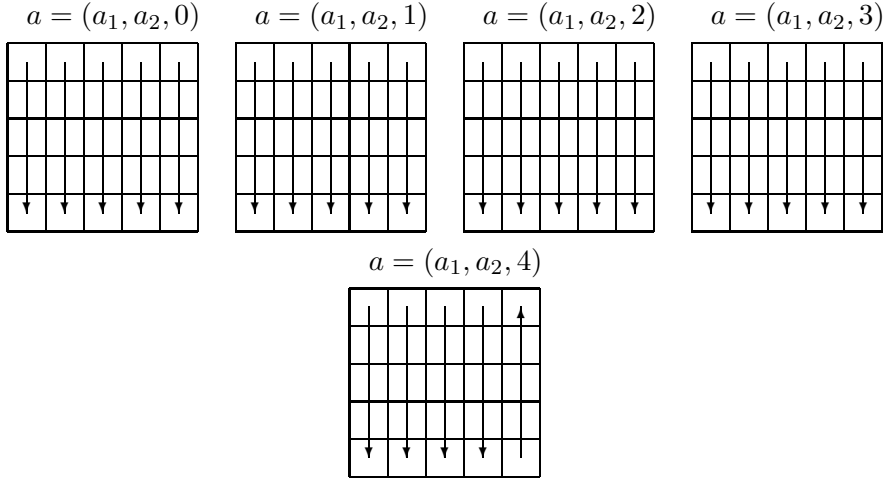
\setlength{\unitlength}{5mm}
\begin{picture}(0,2)
\put(0.5,0.5){$a=(a_1,a_2,0)$} \put(6.5,0.5){$a=(a_1,a_2,1)$}
\put(12.5,0.5){$a=(a_1,a_2,2)$} \put(18.5,0.5){$a=(a_1,a_2,3)$}
\put(9.5,-6){$a=(a_1,a_2,4)$}

\multiput(0,0)(0,-1){6}{\line(1,0){5}}
\multiput(0,0)(1,0){6}{\line(0,-1){5}}
\multiput(0.5,-0.5)(1,0){5}{\vector(0,-1){4}}

\multiput(6,0)(0,-1){6}{\line(1,0){5}}
\multiput(6,0)(1,0){6}{\line(0,-1){5}}
\multiput(6.5,-0.5)(1,0){5}{\vector(0,-1){4}}

\multiput(12,0)(0,-1){6}{\line(1,0){5}}
\multiput(12,0)(1,0){6}{\line(0,-1){5}}
\multiput(12.5,-0.5)(1,0){5}{\vector(0,-1){4}}

\multiput(18,0)(0,-1){6}{\line(1,0){5}}
\multiput(18,0)(1,0){6}{\line(0,-1){5}}
\multiput(18.5,-0.5)(1,0){5}{\vector(0,-1){4}}

\multiput(9,-6.5)(0,-1){6}{\line(1,0){5}}
\multiput(9,-6.5)(1,0){6}{\line(0,-1){5}}
\multiput(9.5,-7)(1,0){4}{\vector(0,-1){4}}\put(13.5,-11){\vector(0,1){4}}
\end{picture}
\vspace{60mm}

\noindent and similarly for $\pi_3$.

Since all of these permutations are products of cycles of odd length
they are all even permutations.

We claim that $G:=\langle \pi_1,\ldots,\pi_n\rangle$ is
2-transitive. It is easy to see that $G$ acts transitively on $A^n$.
We will show that the stabilizer of the state $(q-1,q-1,\ldots,q-1)$ is
transitive on the remaining states.

Note that all but one cycle of $\pi_1$ has length $q$, which is
coprime to 3, the remaining cycle being a 3-cycle. It follows that
$\tau:=\pi_1^q$ is a 3-cycle. It is easy to see that repeatedly
conjugating $\tau$ by the various $\pi_i$s we have enough elements
for the stabilizer of $(q-1,q-1,\ldots,q-1)$ to act transitively on the
remaining points, that is, $G$ acts 2-transitively.

Since any 2-transitive action is primitive it follows that $G$ acts
primitively and since $\tau\in G$ is a 3-cycle, Jordan's Theorem
tells us that $G=\Alt(A^n)$.
\end{proof}

\begin{lemma}
If $q\equiv$ 0 or 2 ($mod$ 6) and $q>2$ then Alt($A^n$) is generated
by $n$ instructions.
\end{lemma}

\begin{proof}
We define the permutation $\pi_1:A^n\rightarrow A^n$ as follows. For
$a:=(a_1,a_2,\ldots,a_n)$ we have that
$$
\pi_1:a\mapsto\left\{\begin{array}{ll}
(1,0,\ldots,0)&\mbox{if }a_1=a_2=\cdots=a_n=0\\
(2,0,\ldots,0)&\mbox{if }a_1=1\mbox{ and }a_2=\cdots=a_n=0\\
(0,0,\ldots,0)&\mbox{if }a_1=2\mbox{ and }a_2=\cdots=a_n=0\\
a&\mbox{if }a_1>2\mbox{ and }a_2=\cdots=a_n=0\\
a&\mbox{if }a_1=0\mbox{ and }a_i\not=0\mbox{ for
some }2\leq i\leq n\\
(1,a_2,\ldots,a_n)&\mbox{if }a_1=q-1\mbox{ and }a_i\not=0\mbox{ for
some }2\leq i\leq n\\
(a_1+1,a_2,\ldots,a_n)&\mbox{otherwise.}\\
\end{array}\right.
$$
For $2\leq r\leq n$ we define the permutation $\pi_r:A^n\rightarrow
A^n$ as follows.
$$
\pi_r:a\mapsto\left\{\begin{array}{ll}
(a_1,\ldots,a_{r-1},0,a_{r+1},\ldots,a_n)&\mbox{if }a_r=q-1\mbox{
and }a_i\not=q-1\mbox{ for some }i\not=r\\
(a_1,\ldots,a_{r-1},a_r+1,a_{r+1},\ldots,a_n)&\mbox{if
}a_r\not=q-1\mbox{ and }a_i\not=q-1\mbox{ for some }i\not=r\\
(a_1,\ldots,a_{r-1},a_r+1,a_{r+1},\ldots,a_n)&\mbox{if
}a_r<q-3\mbox{
and }a_i=q-1\mbox{ for all }i\not=r\\
(a_1,\ldots,a_{r-1},0,a_{r+1},\ldots,a_n)&\mbox{if }a_r=q-3\mbox{
and }a_i=q-1\mbox{ for all }i\not=r\\
a&\mbox{otherwise.}\\
\end{array}\right.
$$
By this stage we believe that the reader has seen sufficiently many
examples and their corresponding diagrams for the reader to be able
draw these themselves.

Since $\pi_1$ is a product of cycles of odd length it is an even
permutation, whilst the permutations $\pi_r$ for $2\leq r\leq n$ are
products of an even number of even cycles and are therefore also
even.

The argument now concludes in the same way as Lemma \ref{AltLemma}.
\end{proof}

\begin{lemma}
If $q\equiv$ 3 ($mod$ 6) then Alt($A^n$) is generated by $n$
instructions.
\end{lemma}

\begin{proof}
We define the permutation $\pi_1:A^n\rightarrow A^n$ as follows. For
$a:=(a_1,a_2,\ldots,a_n)$ we have that
$$
\pi_1:a\mapsto\left\{\begin{array}{ll}
(1,0,\ldots,0)&\mbox{if }a_1=a_2=\cdots=a_n=0\\
(2,0,\ldots,0)&\mbox{if }a_1=1\mbox{ and }a_2=\cdots=a_n=0\\
(0,0,\ldots,0)&\mbox{if }a_1=2\mbox{ and }a_2=\cdots=a_n=0\\
a&\mbox{if }a_1>2\mbox{ and }a_2=\cdots=a_n=0\\
a&\mbox{if }a_1\in\{0,1\}\mbox{ and }a_i\not=0\mbox{ for
some }2\leq i\leq n\\
(2,a_2,\ldots,a_n)&\mbox{if }a_1=q-1\mbox{ and }a_i\not=0\mbox{ for
some }2\leq i\leq n\\
(a_1+1,a_2,\ldots,a_n)&\mbox{otherwise.}\\
\end{array}\right.
$$
For $2\leq r\leq n$ we define the permutation $\pi_r:A^n\rightarrow
A^n$ as follows.
$$
\pi_r:a\mapsto\left\{\begin{array}{ll}
(a_1,\ldots,a_{r-1},0,a_{r+1},\ldots,a_n)&\mbox{if }a_r=q-1\mbox{
and }a_i\not=q-1\mbox{ for some }i\not=r\\
(a_1,\ldots,a_{r-1},a_r+1,a_{r+1},\ldots,a_n)&\mbox{if
}a_r\not=q-1\mbox{ and }a_i\not=q-1\mbox{ for some }i\not=r\\
(a_1,\ldots,a_{r-1},q-1,a_{r+1},\ldots,a_n)&\mbox{if }a_r=0\mbox{
and }a_i=q-1\mbox{ for all }i\not=r\\
(a_1,\ldots,a_{r-1},a_r-1,a_{r+1},\ldots,a_n)&\mbox{otherwise.}\\
\end{array}\right.
$$
Since all of these are products of cycles of odd length they are all
even permutations.

The argument now concludes in the same way as Lemma \ref{AltLemma}.
\end{proof}

\begin{lemma}
If $q=4$ then Alt($A^n$) is generated by $n$ instructions.
\end{lemma}

\begin{proof}
We define the permutation $\pi_1:A^n\rightarrow A^n$ as follows. For
$a:=(a_1,a_2,\ldots,a_n)$ we have that
$$
\pi_1:a\mapsto\left\{\begin{array}{ll}
(1,0,\ldots,0)&\mbox{if }a_1=0\mbox{ and }a_2=a_3=\cdots=a_n=0\\
(2,0,\ldots,0)&\mbox{if }a_1=1\mbox{ and }a_2=a_3=\cdots=a_n=0\\
(0,0,\ldots,0)&\mbox{if }a_1=2\mbox{ and }a_2=a_3=\cdots=a_n=0\\
a&\mbox{if }a_1=3\mbox{ and }a_2=a_3=\cdots=a_n=0\\
(4-a_1,a_2,\ldots,a_n)&\mbox{otherwise.}\\
\end{array}\right.
$$
For $2\leq r\leq n$ we define the permutation $\pi_r:A^n\rightarrow
A^n$ as follows.
$$
\pi_r:a\mapsto\left\{\begin{array}{ll}
(a_1,\ldots,a_{r-1},0,a_{r+1},\ldots,a_n)&\mbox{if }a_r=3\mbox{
and }a_i\not=3\mbox{ for some }i\not=r\\
(a_1,\ldots,a_{r-1},a_r+1,a_{r+1},\ldots,a_n)&\mbox{if
}a_r\not=3\mbox{ and }a_i\not=3\mbox{ for some }i\not=r\\
(a_1,\ldots,a_{r-1},1-a_r,a_{r+1},\ldots,a_n)&\mbox{if
}a_r\in\{0,1\}\mbox{
and }a_i=3\mbox{ for all }i\not=r\\
a&\mbox{otherwise.}\\
\end{array}\right.
$$

Since $\pi_1$ is a product of a three cycle and an even number of
disjoint transpositions it is an even permutation. For $2\leq r\leq
n$ we have that $\pi_r$ is a product of an odd number of cycles of
length 4 and a transposition, so they are also even permutations.

The argument now concludes in the same way as Lemma \ref{AltLemma}.
\end{proof}

\begin{lemma}
If $q\equiv$ 4 ($mod$ 6) and $q>4$ then Alt($A^n$) is generated by
$n$ instructions.
\end{lemma}

\begin{proof}
We define the permutation $\pi_1:A^n\rightarrow A^n$ as follows. For
$a:=(a_1,a_2,\ldots,a_n)$ we have that
$$
\pi_1:a\mapsto\left\{\begin{array}{ll}
(1-a_1,a_2,\ldots,a_n)&\mbox{if }a_1\in\{0,1\}\mbox{ and
}a_2=a_3=\cdots=a_n=0\\
(4,0,\ldots,0)&\mbox{if }a_1=3\mbox{ and }a_2=a_3=\cdots=a_n=0\\
(5,0,\ldots,0)&\mbox{if }a_1=4\mbox{ and }a_2=a_3=\cdots=a_n=0\\
(3,0,\ldots,0)&\mbox{if }a_1=5\mbox{ and }a_2=a_3=\cdots=a_n=0\\
a&\mbox{if }a_1>5\mbox{ and }a_2=\cdots=a_n=0\\
(0,a_2,\ldots,a_n)&\mbox{if }a_1=q-1\mbox{ and }a_i\not=0\mbox{ for
some }2\leq i\leq n\\
(a_1+1,a_2,\ldots,a_n)&\mbox{otherwise.}\\
\end{array}\right.
$$
For $2\leq r\leq n$ we define the permutation $\pi_r:A^n\rightarrow
A^n$ as follows.
$$
\pi_r:a\mapsto\left\{\begin{array}{ll}
(a_1,\ldots,a_{r-1},0,a_{r+1},\ldots,a_n)&\mbox{if }a_r=q-1\mbox{
and }a_i\not=q-1\mbox{ for some }i\not=r\\
(a_1,\ldots,a_{r-1},a_r+1,a_{r+1},\ldots,a_n)&\mbox{if
}a_r\not=q-1\mbox{ and }a_i\not=q-1\mbox{ for some }i\not=r\\
(a_1,\ldots,a_{r-1},q-1,a_{r+1},\ldots,a_n)&\mbox{if }a_r=0\mbox{
and }a_i=q-1\mbox{ for all }i\not=r\\
(a_1,\ldots,a_{r-1},a_r+1,a_{r+1},\ldots,a_n)&\mbox{if
}a_r<q-3\mbox{
and }a_i=q-1\mbox{ for all }i\not=r\\
(a_1,\ldots,a_{r-1},0,a_{r+1},\ldots,a_n)&\mbox{if }a_r=q-1\mbox{
and }a_i=q-3\mbox{ for all }i\not=r\\
a&\mbox{otherwise.}\\
\end{array}\right.
$$

Since $\pi_1$ is a product of a cycle of length 3, an odd number of
cycles of length $q$ and a transposition it is an even permutation.
For $2\leq r\leq n$ we have that $\pi_r$ is a product of even cycles
and is therefore an even permutation.

The argument now concludes in the same way as Lemma \ref{AltLemma}.
\end{proof}

\section{$l$-ary instructions} \label{sec:l-ary}

Theorem 7 in \cite{GR11a} shows that the set of binary instructions only generates the affine group if $A = \GF(2)$. We characterise in Theorem \ref{th:l-ary} below what $l$-ary instructions generate in general.

\begin{theorem} \label{th:l-ary}
For all $2 \le l \le n$, let $G_l$ be the group generated by all $l$-ary instructions. We have the following:
\begin{itemize}
    \item For any $n$, $G_n = \Sym(A^n)$.

    \item If $q=2$, $G_2 = \Aff(n,2)$ for all $n$ and $G_3 = G_{n-1} = \Alt(A^n)$ for $n \ge 4$.

    \item If $n \ge 3$ and $q \ge 3$ is odd, then $G_2 = G_{n-1} = \Sym(A^n)$.

    \item If $n \ge 3$ and $q \ge 4$ is even, then $G_2 = G_{n-1} = \Alt(A^n)$.
\end{itemize}
\end{theorem}

\begin{proof}
Note that, any $(n-1)$-ary instruction has the property
that all its cycle lengths have multiplicities which are multiples of
$q$, since $1$ variable has no effect. So if $q$ is even,
every instruction is an even permutation, and the group $G_{n-1}$ is contained in the
alternating group.

When $q=2$, $G_2$ was settled in \cite[Theorem 5.8]{GR11a} and it directly follows from \cite{KM74} that $G_3 = \Alt(\GF(2)^n)$.

We now assume $q \ge 3$ and it suffices to prove the result for $G_2$.

We use the following principles. Let $G$ be a permutation group.
(a) Define a relation $\sim$ by $x\sim y$ if $x=y$ or the transposition
$(x,y)$ is in $G$. Then $\sim$ is an equivalence relation. If it is the
universal relation, then $G$ contains every transposition, and is the
symmetric group.
(b) Define a relation $\sim_1$ by $x\sim_1 y$ if $x=y$ or there exists $z$ such
that the $3$-cycle $(x,y,z)$ is in $G$. Again $\sim_1$ is an equivalence
relation, and it has the property that, if $x\sim_1 y$ and $y\sim_1 z$, then
$(x,y,z)\in G$. So if $\sim_1$ is the universal relation, then $G$ contains
every $3$-cycle, and is the alternating group.

Consider $n=2$. There are two kinds of instruction, those that update $y_1$
and those that update $y_2$. Thinking of the points being permuted as forming
a square grid, instructions of the first time form a group which is the
direct product of symmetric groups on the columns; instructions of the second
type form a group which is the direct product of symmetric groups on the
rows. So the relation $\sim$ is non-trivial, and any two points in the same
row or column are equivalent. Thus $\sim$ is the universal relation, and $G$
is the symmetric group.

Now suppose that $n=3$ and $q\ge3$. By the $n=2$ case, we see that every
permutation fixing
the first coordinate of all triples is in $G$. In particular, the permutation
transposing $(x,1,1)$ with $(x,1,2)$ for all $x$ belongs to $G$. Similarly
the permutation transposing $(1,y,1)$ and $(1,y,0)$ for all $y$ is in $G$.
(Here we use $q\ge3$.) Now in these permutations, the cycles containing
$(1,1,1)$ intersect; the other cycles are disjoint. So the commutator of the
two permutations is a $3$-cycle. Now applying the argument about $3$-cycles
shows that $G$ is symmetric or alternating. If $q$ is even, it is alternating;
if $q$ is odd, since we have a product of $q$ transpositions (an odd
permutation), $G$ is symmetric.

Finally, assume $n\ge4$ and $q \ge 5$. Argue as before but using $3$-cycles rather than
transpositions; the condition $q\ge5$ allows us to have $3$-cycles meeting
in a single point, so their commutator is a $3$-cycle.

For $q=3$, by induction we get the symmetric group at the previous stage,
and so we have transpositions, and can play the usual game.

For $q=4$, the commutator trick gives us a permutation $t$ which interchanges, say,
$(1,\ldots,1,1,1)$ with $(1,\ldots,1,1,2)$, and $(1,\ldots,1,1,3)$ with
$(1,\ldots,1,1,0)$. Now there is an instruction $g$ involving the last two
coordinates which fixes three of these points and maps the fourth to
$(1,\ldots,1,2,0)$. Conjugating $t$ by $g$ gives a permutation which is
the product of two transpositions, the first the same as in $t$, the
second swapping $(1,\ldots,1,1,3)$ with $(1,\ldots,1,2,0)$. Now the
commutator of $t$ and $t^g$ is a $3$-cycle on $(1,\ldots,1,1,3)$,
$(1,\ldots,1,1,0)$ and $(1,\ldots,1,2,0)$. So the relation $\sim_1$
is non-trivial, and points which differ in only one coordinate are
equivalent. So it is the universal relation, and we are done.
\end{proof}

We would like to emphasise the importance of Theorem \ref{th:l-ary} for the possible implementation of memoryless computation. Recall that any boolean function $f: \{0,1\}^n \to \{0,1\}$ can be computed by using binary gates (NAND suffices, in fact). This shows that any function can be computed ``locally''. However, when we want to compute an odd boolean permutation $g: \{0,1\}^n \to \{0,1\}^n$ without memory, Theorem \ref{th:l-ary} implies that we need an $n$-ary instruction to compute it. In other words, $g$ cannot be computed ``locally''.

The main impact of this result (more specifically, its generalisation to any alphabet of even cardinality) is about the design of instruction sets. A CPU core cannot perform any possible operation on all its registers. The typical approach is to use assembly instructions which only work on two or three registers at once (i.e., binary or ternary instructions). However, in any implementation of a core which computes without memory, the instruction set must contain at least one instruction which uses all registers at once. The typical approach to the design of instruction sets is thus inappropriate to memoryless computation.

\section{Perspectives} \label{sec:open}

First of all, the reader is reminded of the conjectures on fast groups in Section \ref{sec:preliminaries}. More generally, memoryless computation, despite some previous work, remains a dawning topic. Many problems, especially of computational nature, arise in this area. For instance, we know that any permutation can be computed by a program with at most $2n-1$ instructions. However, how hard is it to determine such a program computing a given permutation? This problem has a linear analogue: any matrix can be computed in $\lfloor 3n/2\rfloor$ instructions \cite{CFG12a}, but the complexity of determining those instructions remains unknown.

Also, in this paper we started the investigation of smaller instruction sets which could still provide the advantages of memoryless computation. In particular, we determined the smallest cardinality of a generating set of instructions. Clearly, using such instruction set will yield very long programs.  Can we derive results on the complexity of permutations when the smallest generating sets of instructions are used? We also studied the use of $l$-ary permutations and showed that they were not sufficient in general. Can we bound the complexity of permutations when binary instructions are sufficient? Furthermore, which other generating sets of instructions could be proposed for memoryless computation, and how good are they in terms of size and complexity?

Finally, an interesting application of our results on memoryless computation is in bioinformatics, more precisely in the modelling of gene regulatory networks, introduced in \cite{Kau69} and \cite{Tho80} (see \cite{KS08} and \cite{PR12} for reviews on this topic). A network of $n$ genes interacting with one another is typically modelled as follows. To each gene $i$ ($1 \le i \le n$) are associated the following:
\begin{itemize}
	\item Firstly, a variable $x_i$, called its state, taking a value in a finite alphabet $A$ of size $q$, which indicates the level of activation of the gene (usually, $q=2$, hence the common notation of a Boolean network).
	\item Secondly, and an update function $f_i : A^k \to A$ which depends on the values of some genes $j_1,\ldots,j_k$ that influence its level of activation: $f_i(x_{j_1}, \ldots, x_{j_k})$.
\end{itemize}	
In general, the order in which the genes update their state, referred to as the update scheme, is unknown. In some models, all the genes are assumed to update their state synchronously, i.e. all at the same time (this is the so-called parallel update scheme). In other cases, the updates can be done asynchronously, in particular, they can be assumed to be updated one after the other (this is the so-called serial update scheme). It is clear that memoryless computation corresponds to the serial update scheme, also called a sequential dynamic system \cite{MR08}, where an update of gene $i$ corresponds to the instruction
$$
	y_i \gets f_i(y_{j_1}, \ldots, y_{j_k}).
$$ 
One major question is to determine whether any generality is lost by considering one kind of update schedule over another. In this paper, we have proved in Theorem \ref{th:sym_alt} (a) the universality of the serial update for permutations. Indeed, for any $A$ and any $n$ (apart from the degenerate case $q = n = 2$), there exists a gene regulatory network $f_1,\ldots,f_n$ which can generate any possible permutation of $A^n$ in its serial update scheme, i.e. for any $g = (g_1,\ldots,g_n) \in \Sym(A^n)$, there exists a word $(i_1,i_2,\ldots, i_L) \in \{1,\ldots,n\}^L$ such that successively updating the states of genes $i_1, i_2,\ldots,i_L$ eventually yields the state $(g_1(x),\ldots,g_n(x))$.



\end{document}